\documentclass[a4paper, 11pt]{article}
\usepackage{amsfonts}
\usepackage{amsthm}
\usepackage{amsmath}
\usepackage{authblk}
\usepackage{bm}
\usepackage{color}
\usepackage{graphicx}
\usepackage{enumitem}
\usepackage{lscape}
\usepackage{mathtools}
\usepackage{url}
\usepackage{hyperref}
\usepackage{microtype}
\usepackage[initials,nobysame]{amsrefs}
\usepackage{cleveref}

\newcommand{\FF}{\mathbb{F}}

\newcommand{\0}{\mathbf{0}}

\DeclareMathOperator{\wt}{wt}
\DeclareMathOperator{\Hull}{Hull}
\DeclareMathOperator{\rank}{rank}

\newenvironment{smpmat}{\left ( \begin{smallmatrix}} {\end{smallmatrix}\right )}

\theoremstyle{plain}
\newtheorem{thm}{Theorem}[section]
\newtheorem{lem}[thm]{Lemma}
\newtheorem{cor}[thm]{Corollary}
\newtheorem{prp}[thm]{Proposition}
\theoremstyle{definition}

\theoremstyle{remark}
\newtheorem{rem}[thm]{Remark}
\crefname{thm}{Theorem}{Theorems}
\crefname{lem}{Lemma}{Lemmas}
\crefname{cor}{Corollary}{Corollaries}
\crefname{prp}{Proposition}{Propositions}
\crefname{section}{Section}{Sections}
\crefname{table}{Table}{Tables}

\makeatletter{}
\makeatother{}

\title{Construction of quaternary Hermitian LCD codes}
\author{Keita Ishizuka\thanks{Corresponding author. Research Center for Pure and Applied Mathematics Graduate School of Information Sciences, Tohoku University, Sendai 980–8579, Japan.\ email: \texttt{keita.ishizuka.p5@dc.tohoku.ac.jp}}}
\date{}

\begin{document}
\maketitle{}
\begin{abstract}
    We introduce a general construction of many Hermitian LCD $[n,k]$ codes from a given Hermitian LCD $[n,k]$ code.
    Furthermore, we present some results on punctured codes and shortened codes of quaternary Hermitian LCD codes.
    As an application, we improve some of the previously known lower bounds on the largest minimum weights of quaternary Hermitian LCD codes of length $12 \le n \le 30$.
\end{abstract}


\section{Introduction} \label{sec:introduction}
Linear complementary dual codes, LCD codes for short, are codes that intersect their dual codes trivially.
Massey~\cite{massey_linear_1992} introduced Euclidean LCD codes and showed that Euclidean LCD codes provide an optimum linear coding solution for the two-user binary adder channel.
Since then, much work has been done concerning both Euclidean LCD codes and Hermitian LCD codes.
Carlet and Guilley~\cite{carlet_complementary_2016} applied binary Euclidean LCD codes in side-channel attacks and fault non-invasive attacks.
Lu, Li, Guo and Fu~\cite{lu_maximal_2015} proved that a quaternary Hermitian LCD code leads to a construction of a maximal-entanglement entanglement-assisted quantum error correcting code.
In particular, Carlet, Mesnager, Tang, Qi and Pellikaan~\cite{carlet_linear_2018} showed that any code over $\FF_q$ is equivalent to some Euclidean LCD code for $q \ge 4$ and any code over $\FF_{q^2}$ is equivalent to some Hermitian LCD code for $q \ge 3$.
This motivates us to study the largest minimum weights of quaternary Hermitian LCD codes.
We say that a quaternary Hermitian LCD $[n,k]$ code is optimal if it has the largest minimum weight among all quaternary Hermitian LCD $[n,k]$ codes.

We discuss the recent research progress on the largest minimum weights of quaternary Hermitian LCD codes as follows.
Let $d_4(n,k)$ denote the largest minimum weight of quaternary Hermitian LCD $[n,k]$ codes.
For $n \in \{1, 2, \dots, 11 \}$ and an arbitrary value of $k$, Lu, Zhan, Yang and Cao~\cite{lu_optimal_2020} determined $d_4(n,k)$.
For an arbitrary value of $n$ and $k=2$, Lu, Li, Guo and Fu~\cite{lu_maximal_2015} determined $d_4(n,k)$.
Furthermore, for an arbitrary value of $n$ and $k=3$, the value of $d_4(n,k)$ were determined by Lu, Li, Guo and Fu~\cite{lu_maximal_2015} and Araya, Harada and Saito~\cite{araya_quaternary_2020}.
Zhan,Li, Lu and Li~\cite{zhan_quatemary_2020} constructed quaternary Hermitian LCD codes with minimum weight $d=3,4,5,6$ of length $n<52$.
Sok~\cite{sok_hermitian_2020} used unitary matrices to construct many quaternary Hermitian LCD codes.

In this paper, we introduce a method for constructing many $[n,k]$ codes from a given $[n,k]$ code which preserves the dimension of the Hermitian hull, extending the result by Ishizuka and Saito~\cite{ishizuka_construction_2022}.
Although we apply the method only to quaternary Hermitian LCD codes, it can be applied to both Hermitian self-dual codes and Hermitian LCD codes over $\FF_{q^2}$ with any prime power $q$.
Furthermore, we provide some theorems concerning punctured codes and shortened codes of quaternary Hermitian LCD codes, motivated by the result of Bouyuklieva~\cite{bouyuklieva_optimal_2021} on binary Euclidean LCD codes.
As an application, we construct quaternary Hermitian LCD codes and improve some of the previously known lower bounds on $d_4(n,k)$ for $12 \le n \le 30$.

This paper is organized as follows:
In \cref{sec:preliminaries}, we give some definitions, notations and basic results.
In \cref{sec:construction}, we introduce the construction method.
In \cref{sec:puncture}, we give some theorems on punctured codes and shortened codes of quaternary Hermitian LCD codes.
In \cref{sec:optimal}, applying the results, we construct quaternary Hermitian LCD codes to improve some of the previously known lower bounds.

\section{Preliminaries} \label{sec:preliminaries}

Let $\FF_q$ be the finite field of order $q$, where $q$ is a prime power and let $\FF_q^n$ be the vector space of all $n$-tuples over $\FF_q$.
Especially, let $\FF_4 = \{0,1,\omega,\omega^2\}$ be the finite field of order four, where $\omega$ satisfies $\omega^2+\omega+1=0$.
An $[n,k]$ code over $\FF_q$ is a $k$-dimensional subspace of $\FF_q^n$.
Codes over $\FF_4$ are said to be quaternary codes.
For an $[n,k]$ code $C$ over $\FF_q$, the parameters $n$ and $k$ are said to be the length and the dimension of $C$, respectively.
The elements of a code are known as codewords.
The weight of $x=(x_1,x_2,\dots,x_n) \in \FF_q^n$ is defined as $\wt(x) = |\{ i \mid x_i \ne 0 \}|$.
The minimum weight of a code $C$ is defined as $\wt(C) = \min \{ \wt(x) \mid x \in C, x \neq \0_n \}$, where $\0_n$ denotes the zero vector of length $n$.
An $[n,k,d]$ code over $\FF_q$ is an $[n,k]$ code with minimum weight $d$.
We say that a code is even if every codeword has even weight.
If a code is not even, then it is said to be odd.

Let $C$ be a code over $\FF_q$ with length $n$ and let $T \subseteq \{ 1, 2, \dots, n \}$.
Deleting every coordinate $i \in T$ in every codeword of $C$ gives a code called the punctured code on $T$ and denoted by $C^T$.
Define a subcode $C(T) = \{ (c_1,c_2,\dots,c_n) \in C \mid c_i=0 \text{ for all } i \in T \}$.
Puncturing $C(T)$ on $T$ gives a code called the shortened code on $T$ and denoted by $C_T$.
If $|T|=1$, say $T=\{i\}$, then we will write $C^{\{i\}}$ and $C_{\{i\}}$ as $C^i$ and $C_i$, respectively.

For $a \in \FF_{q^2}$, the conjugate of $a$ is defined as $\overline{a} = a^q$.
The Hermitian inner product of $x=(x_1,x_2,\dots,x_n), y=(y_1,y_2,\dots,y_n) \in \FF_{q^2}^n$ is defined as $(x,y)_h = \sum_{i=1}^n x_i \overline{y_i}$.
For any $[n,k]$ code $C$ over $\FF_{q^2}$, the Hermitian dual code of $C$ is defined as $C^{\perp h} = \{ x \in \FF_{q^2}^n \mid (x, y)_h = 0 \text{ for all } y \in C \}$.
Given a code $C$, we denote by $d^{\perp h}$ the minimum weight of $C^{\perp h}$.
A code $C$ over $\FF_{q^2}$ is said to be a Hermitian LCD code if $C \cap C^{\perp h} = \{ \0_n \}$.
We say that a Hermitian LCD $[n,k]$ code over $\FF_{q^2}$ is optimal if it has the largest minimum weight among all Hermitian LCD $[n,k]$ codes over $\FF_{q^2}$.
A code $C$ over $\FF_{q^2}$ is said to be a Hermitian self-orthogonal code if $C \subseteq C^{\perp h}$.
If $C$ satisfies $C=C^{\perp h}$, then we say that $C$ is a Hermitian self-dual code.
The Hermitian hull of a code $C$ is defined as $\Hull_h(C) = C \cap C^{\perp h}$.
By definition, Hermitian self-dual codes and Hermitian LCD codes are maximum and minimum with respect to the dimensions of the Hermitian hulls, respectively.
The construction method introduced in this paper preserves the dimension of the Hermitian hull.
Thus, it can be applied to constructions of Hermitian self-dual codes and Hermitian LCD codes over $\FF_{q^2}$ for any prime power $q$.
In this paper, however, we apply the method only to quaternary Hermitian LCD codes.

A generator matrix of a code $C$ is any matrix whose rows form a basis of $C$.
For a matrix $A$, we denote by $A^T$ and $\overline{A}$ the transpose and the conjugate of $A$, respectively.
The following results will be important in \cref{sec:construction,sec:puncture}.
\begin{lem}[{\cite[Proposition 3.5]{guneri_quasi-cyclic_2016}}] \label{thm:guneri2016}
    Let $C$ be an $[n,k]$ code over $\FF_{q^2}$ with generator matrix $G$.
    Then $C$ is a Hermitian LCD code if and only if the $k \times k$ matrix $G \overline{G}^T$ is nonsingular.
\end{lem}
\begin{lem} [{\cite[Theorem 1]{macwilliams_self-dual_1978}}] \label{thm:macwilliams1978}
    A quaternary code is Hermitian self-orthogonal if and only if the code is an even code.
\end{lem}
\begin{lem}[{\cite[Proposition 3.3]{guenda_constructions_2018}}] \label{thm:guenda2018}
    Let $C$ be an $[n,k]$ code over $\FF_{q^2}$ with generator matrix $G$.
    Then
    \begin{equation*}
        \rank(G\overline{G}^T) = k - \dim(\Hull_h(C)).
    \end{equation*}
\end{lem}

\section{Construction method} \label{sec:construction}
Recently, Harada~\cite{harada_construction_2021} introduced a method to construct many Euclidean LCD $[n,k]$ codes from a given Euclidean LCD $[n,k]$ code, modifying a known method for Euclidean self-dual codes~\cite{harada_existence_1996}.
Extending the result, Ishizuka and Saito~\cite{ishizuka_construction_2022} provided a method to construct many $[n,k]$ code from a given $[n,k]$ code, preserving the dimension of the Euclidean hull.
In this section, we give a construction method similar to that of~\cite{ishizuka_construction_2022}, which preserves the dimension of the Hermitian hull.

Let $C$ be an $[n,k]$ code over $\FF_{q^2}$ with generator matrix $\begin{pmatrix} I_k & A \\ \end{pmatrix}$, where $I_k$ denotes the identity matrix of order $k$ and $A$ is a $k \times (n-k)$ matrix.
We denote by $r_i$ the $i$-th row of $A$.
For $A$ and $x,y \in \FF_{q^2}^{n-k}$, define an $k \times (n-k)$ matrix $A(x,y)$, where the $i$-th row $r'_i$ is defined as follows:
\begin{equation*}
    r'_i = r_i + (r_i,y)_h x - (r_i,x)_h y. \label{eq:ri}
\end{equation*}
We denote by $C(A(x,y))$ the code with generator matrix $\begin{pmatrix} I_k & A(x,y) \\ \end{pmatrix}$.
\begin{rem}
    With the above notation, suppose that $x=\0_{n-k}$ or $y=\0_{n-k}$.
    Then it holds that $A(x,y)=A$.
    Hereafter, we assume that $x \neq \0_{n-k}$ and $y \neq \0_{n-k}$.
\end{rem}
\begin{thm} \label{thm:hulldim}
    Let $C$ be an $[n,k]$ code over $\FF_{q^2}$ with generator matrix $G = \begin{pmatrix} I_k & A \\ \end{pmatrix}$ and let $x,y \in \FF_{q^2}^{n-k}$.
    Suppose that $(x,x)_h=(y,y)_h=(x,y)_h=0$.
    Then $\dim(\Hull_h(C(A(x,y)))) = \dim(\Hull_h(C))$.
\end{thm}
\begin{proof}
    We denote by $r_i$ and $r'_i$ the $i$-th rows of $A$ and $A(x,y)$, respectively.
    It holds that
    \begin{equation*}
        \begin{split}
            (r_i',r_j')_h
            = &(r_i + (r_i,y)_h x - (r_i,x)_h y,\ r_j + (r_j,y)_h x - (r_j,x)_h y)_h       \\
            = &(r_i,r_j)_h + \overline{(r_j,y)_h}(r_i,x)_h - \overline{(r_j,x)_h}(r_i,y)_h \\
            &+ (r_i,y)_h(x,r_j)_h - (r_i,x)_h(y,r_j)_h                                       \\
            = &(r_i,r_j)_h.
        \end{split}
    \end{equation*}
    Therefore, it follows that
    \begin{equation*} \label{eq:LCP}
        \begin{split}
            \begin{pmatrix} I_k & A(x,y) \\ \end{pmatrix}
            \overline{\begin{pmatrix} I_k & A(x,y) \\ \end{pmatrix}}^T
            &= I_k + A(x,y)\overline{A(x,y)}^T\\
            &= I_k + A\overline{A}^T\\
            &= \begin{pmatrix} I_k & A \\ \end{pmatrix} \overline{\begin{pmatrix} I_k & A \\ \end{pmatrix}}^T.
        \end{split}
    \end{equation*}
    By \cref{thm:guenda2018}, the result follows.
\end{proof}

\Cref{thm:hulldim} states that the construction method preserves the dimension of the Hermitian hull.
Recall that Hermitian self-dual codes and Hermitian LCD codes are maximum and minimum with respect to the dimensions of the Hermitian hulls, respectively.
Therefore, we obtain \cref{thm:construction_LCD,thm:construction_self-dual}.
\begin{cor} \label{thm:construction_LCD}
    Let $C$ be an $[n,k]$ code over $\FF_{q^2}$ with generator matrix $G = \begin{pmatrix} I_k & A \\ \end{pmatrix}$ and let $x,y \in \FF_{q^2}^{n-k}$.
    Suppose that $(x,x)_h=(y,y)_h=(x,y)_h=0$.
    Then $C(A(x,y))$ is a Hermitian LCD code if and only if $C$ is a Hermitian LCD code.
\end{cor}
\begin{cor} \label{thm:construction_self-dual}
    Let $C$ be an $[n,k]$ code over $\FF_{q^2}$ with generator matrix $G = \begin{pmatrix} I_k & A \\ \end{pmatrix}$ and let $x,y \in \FF_{q^2}^{n-k}$.
    Suppose that $(x,x)_h=(y,y)_h=(x,y)_h=0$.
    Then $C(A(x,y))$ is a Hermitian self-dual code if and only if $C$ is a Hermitian self-dual code.
\end{cor}

\section{Punctured codes and shortened codes of quaternary Hermitian LCD codes} \label{sec:puncture}

Recently, Bouyuklieva~\cite{bouyuklieva_optimal_2021} gave a characterization of both punctured codes and shortened codes of binary LCD codes.
In this section, we show that similar results hold for quaternary Hermitian LCD codes.

For a code $C$ of length $n$ and a set $T \subseteq \{1,2,\dots,n\}$, define a subcode $C(T) = \{(c_1,c_2,\dots,c_n) \in C \mid c_i = 0 \text{ for all } i \in T \}$.
For Euclidean dual codes, a result similar to~\cref{thm:PperpSperp} is given in~\cite[Theorem 1.5.7]{huffman_fundamentals_2003}.
By an argument similar to that in the proof of~\cite[Theorem 1.5.7]{huffman_fundamentals_2003}, we obtain \cref{thm:PperpSperp}.

\begin{lem} \label{thm:PperpSperp}
    Let $C$ be a code of length $n$ over $\FF_{q^2}$ and $T \subseteq \{1,2,\dots,n\}$.
    Then $(C^{\perp h})_T=(C^T)^{\perp h}$ and $(C^{\perp h})^T=(C_T)^{\perp h}$.
\end{lem}
\begin{proof}
    Let $c \in C^{\perp h}(T)$ and let $c^*$ be the codeword obtained from $c$ by removing the coordinates in $T$.
    If $x \in C$, then $0=(x,c)_h=(x^*,c^*)_h$, where $x^*$ is the codeword $x$ punctured on $T$.
    Hence $(C^{\perp h})_T \subset (C^T)^{\perp h}$.
    Any vector $c \in (C^T)^{\perp h}$ can be extended to a vector $\hat{c}$ by inserting $0$s in the positions of $T$.
    If $x \in C$, then puncture $x$ on $T$ to obtain $x^*$.
    Since $(\hat{c},x)_h=(c,x^*)_h=0$, it follows that $c \in (C^{\perp h})_T$.
    Hence $ (C^{\perp h})_T = (C^T)^{\perp h} $.
    Replacing $C$ by $C^{\perp h}$ gives $(C^{\perp h})^T = (C_T)^{\perp h}$.
\end{proof}

\begin{thm} \label{thm:P_XOR_S}
    Let $C$ be a quaternary Hermitian LCD $[n,k,d]$ code with $d,d^{\perp h} \ge 2$.
    For all $1 \le i \le n$, exactly one of $C^i$ and $C_i$ is a Hermitian LCD code.
\end{thm}
\begin{proof}
    Without loss of generality, we may assume that $i=1$.
    Since $C$ is a Hermitian LCD code, any vector in $\FF_4^n$ can be written uniquely as the sum of a vector in $C$ and a vector in $C^{\perp h}$.
    Especially, one of the following holds:
    \begin{enumerate}
        \item
              $(1,0,\dots,0)=(\omega,v)+(\omega^2,v)$, where $(\omega,v) \in C$ or $(\omega,v) \in C^{\perp h}$.
        \item
              $(1,0,\dots,0)=(0,v)+(1,v)$, where $(0,v) \in C$ or $(0,v) \in C^{\perp h}$.
    \end{enumerate}
    Suppose that $(1,0,\dots,0)=(\omega,v)+(\omega^2,v)$.
    If $(\omega,v) \in C$, then $(\omega^2,v) \in C^{\perp h}$.
    Hence $((\omega,v), (\omega^2,v))_h=0$.
    However, this is impossible because $((\omega,v),(\omega^2,v))_h=\omega^2+(v,v)_h \ne 0$.
    By the same argument, it is verified that $(\omega,v) \in C^{\perp h}$ is not possible.
    Therefore, we only consider (ii).

    First, consider the case where $(0,v) \in C$.
    It follows that $v \in C_1$ and $v \in (C^{\perp h})^1$.
    By \cref{thm:PperpSperp}, $v \in C_1 \cap (C_1)^{\perp h}$.
    Therefore, $C_1$ is not a Hermitian LCD code.
    Note that $v$ is a nonzero vector since $d,d^{\perp h} \ge 2$.
    Suppose that $C^1$ is not a Hermitian LCD code either.
    Then there exists a nonzero vector $u \in C^1 \cap (C^1)^{\perp h}$.
    From $u \in (C^1)^{\perp h} = (C^{\perp h})_1$, it follows that $(0,u) \in C^{\perp h}$.
    From $u \in C^1$, one of the following holds: $(0,u), (1,u), (\omega,u), (\omega^2,u) \in C$.
    \begin{enumerate}
        \item
              Let $(0,u) \in C$.
              Then it holds that $(0,u) \in C \cap C^{\perp h}$.
              This is a contradiction as $C$ is a Hermitian LCD code.
        \item
              Let $(1,u) \in C$.
              It holds that $(1,0) = (0,v) + (1,v) = (0,u) + (1,u)$.
              Note that $v \ne u$ since $(0,v),(1,u) \in C$ and $d \ge 2$.
              This is a contradiction as $C$ is a Hermitian LCD code.
        \item
              Let $(\omega,u) \in C$.
              Then it follows that $(1, \omega^2 u) \in C$.
              Therefore, $(1,0) = (0,v) + (1,v) = (0, \omega^2 u) + (1, \omega^2 u)$.
              Note that $v \ne \omega^2 u$ since $(0,v),(\omega,u) \in C$ and $d \ge 2$.
              This is a contradiction as $C$ is a Hermitian LCD code.
        \item
              Then it follows that $(1, \omega u) \in C$.
              Therefore, $(1,0) = (0,v) + (1,v) = (0, \omega u) + (1, \omega u)$.
              Note that $v \ne \omega u$ since $(0,v),(\omega^2,u) \in C$ and $d \ge 2$.
              This is a contradiction as $C$ is a Hermitian LCD code.
    \end{enumerate}
    Therefore, it holds that $C^1$ is a Hermitian LCD code.

    Consider the case where $(1,v) \in C$.
    It follows that $v \in C^1$ and $v \in (C^{\perp h})_1$.
    By the same argument as above, it holds that $C^1$ is not a Hermitian LCD code and $C_1$ is a Hermitian LCD code.
\end{proof}

\Cref{thm:orthnormal_basis} is obtained by Ken Saito (private communication, Dec. 31, 2021).
In~\cite{carlet_new_2019}, a similar result was given for binary Euclidean LCD codes.
\begin{lem}[Ken Saito] \label{thm:orthnormal_basis}
    Let $C$ be a quaternary Hermitian LCD $[n,k]$ code.
    Then there exists a generator matrix $G$ of $C$ such that $G\overline{G}^T=I_k$.
\end{lem}
\begin{proof}
    Suppose that $C$ is an even code.
    Then, by Lemma~\ref{thm:macwilliams1978}, $C$ is a Hermitian self-dual code.
    This is a contradiction as $C$ is a Hermitian LCD code.
    Thus, we assume that $C$ is an odd code.
    Then there exists a codeword which has an odd weight.
    From the proof of~\cite[Proposition 4]{harada_remark_2020}, $C$ has a generator matrix of the form
    \begin{math}
        G_0 =
        \begin{pmatrix}
            x_1 \\
            G_1 \\
        \end{pmatrix}
    \end{math},
    satisfying that $(x_1,x_1)_h \ne 0$ and $x_1 \overline{G_1}^T = \0_{k-1}$.
    The code $C_1$ with generator matrix $G_1$ is a Hermitian LCD $[n,k-1]$ code.
    After finitely many steps, we obtain a generator matrix $G$ of $C$:
    \begin{equation*}
        G =
        \begin{pmatrix}
            x_1    \\
            x_2    \\
            \vdots \\
            x_k    \\
        \end{pmatrix},
    \end{equation*}
    where $(x_i,x_j)_h=\delta_{i,j}$ for all $1 \le i,j \le n$.
    Here $\delta_{i,j}$ is the Kronecker delta function.
    Hence $G \overline{G}^T = I_k$.
    This completes the proof.
\end{proof}

\begin{thm} \label{thm:shortenLCD}
    Let $C$ be a quaternary Hermitian LCD $[n,k]$ code and let $G$ be a generator matrix of $C$ such that $G\overline{G}^T=I_k$.
    Let $l_i$ denote the $i$-th column of $G$.
    Then $C^i$ is a Hermitian LCD code if and only if $\wt(l_i)$ is even.
\end{thm}
\begin{proof}
    Without loss of generality, we may assume that
    \[
        l_i = (0,\dots,0,1,\dots,1,\omega,\dots,\omega,\omega^2,\dots,\omega^2)^T.
    \]
    Let $p,q$ and $r$ be the numbers of coordinates whose entries equal to $1,\omega$ and $\omega^2$, respectively.
    Let $G^i$ be a generator matrix of $C^i$ obtained by deleting $i$-th column of $G$.
    Since $G\overline{G}^T = I_k$, it holds that
    \[
        G^i \overline{G^i} = \begin{pmatrix}
            I_{n-p-q-r} & O          & O          & O          \\
            O           & J-I_p      & \omega^2 J & \omega J   \\
            O           & \omega J   & J-I_q      & \omega^2 J \\
            O           & \omega^2 J & \omega J   & J-I_r      \\
        \end{pmatrix},
    \]
    where $J$ and $O$ denote the all-ones matrices and the zero matrix of suitable sizes, respectively.
    Let $M$ be
    \[
        M = \begin{pmatrix}
            J-I_p      & \omega^2 J & \omega J   \\
            \omega J   & J-I_q      & \omega^2 J \\
            \omega^2 J & \omega J   & J-I_r      \\
        \end{pmatrix}.
    \]
    Then $M$ is a submatrix of $G^i \overline{G^i}$ and $\det{G^i \overline{G^i}} = \det{M}$.
    \begin{enumerate}
        \item
              Suppose $p \equiv q \equiv r \equiv 0 \pmod 2$.
              The value of $\det{M}$ is calculated as follows:
              Let $M_1$ denote the matrix whose rows consist of the $i$-th rows of $M$ for $1 \le i \le p$.
              First, add the $i$-th rows of $M_1$ to the first row for $i=2,\dots,p$.
              Adding the first row to $i$-th rows for $i=2,\dots,p$, we obtain $M_1 = \begin{pmatrix} A & B & C \end{pmatrix}$,
              where
              \begin{align*}
                   & A = \begin{pmatrix}
                      1      & 1      & \dots  & 1      \\
                      0      & 1      & \dots  & 0      \\
                      \vdots & \vdots & \ddots & \vdots \\
                      0      & 0      & \dots  & 1
                  \end{pmatrix}, \\
                   & B = \begin{pmatrix}
                      0        & \dots  & 0        \\
                      \omega^2 & \dots  & \omega^2 \\
                               & \vdots &          \\
                      \omega^2 & \dots  & \omega^2 \\
                  \end{pmatrix}, \\
                   & C = \begin{pmatrix}
                      0      & \dots  & 0      \\
                      \omega & \dots  & \omega \\
                             & \vdots &        \\
                      \omega & \dots  & \omega \\
                  \end{pmatrix}.
              \end{align*}
              By the same row operations, we obtain
              \[
                  \det M = \det \begin{pmatrix}
                      A & B & C \\
                      C & A & B \\
                      B & C & A
                  \end{pmatrix}.
              \]
              Add the first row of $\det M$ multiplied by $\omega$ and $\omega^2$ to $i$-th rows for $i=p+2,\dots,p+q$ and $p+q+2,\dots,p+q+r$, respectively.
              Then add the $(p+1)$-st row multiplied by $\omega^2$ and $\omega$ to $i$-th rows for $i=2,\dots,p$ and $p+q+2,\dots,p+q+r$.
              Finally, adding the $(p+q+1)$-st row multiplied by $\omega$ and $\omega^2$ to $i$-th rows for $i=2,\dots,p$ and $p+2,\dots,p+q$, respectively, we obtain
              \[
                  \det M = \det \begin{pmatrix}
                      A & O & O \\
                      O & A & O \\
                      O & O & A
                  \end{pmatrix}.
              \]
              Therefore, it follows that $\det{M}=1$.
        \item
              Suppose that $p \equiv q \equiv r \equiv 1 \pmod 2$.
              The value of $\det{M}$ is calculated as follows:
              By the same row operations as in (i), we obtain
              \[
                  \det M = \det \begin{pmatrix}
                      J-A          & \omega^2 J-B & \omega J-C   \\
                      \omega J-C   & J-A          & \omega^2 J-B \\
                      \omega^2 J-B & \omega J-C   & J-A
                  \end{pmatrix}.
              \]
              Add the first row of $\det M$ to the $(p+q+1)$-st row.
              Then add the $(p+1)$-st row of $\det M$ to the first row.
              Multiplying the first row by $\omega$, we verify that the first row and $(p+q+1)$-st row is identical.
              Therefore, it follows that $\det M = 0$.
        \item
              Suppose that $p \equiv 0 \pmod 2$, $q \equiv r \equiv 1 \pmod 2$.
              The value of $\det{M}$ is calculated as follows:
              By the same row operations as in (i), we obtain
              \[
                  \det M = \det \begin{pmatrix}
                      A            & B          & C            \\
                      \omega J-C   & J-A        & \omega^2 J-B \\
                      \omega^2 J-B & \omega J-C & J-A
                  \end{pmatrix}.
              \]
              Add the first row of $\det M$ multiplied by $\omega$ to the $(p+1)$-st row.
              Furthermore, add the first row multiplied by $\omega^2$ to the $(p+q+1)$-st row.
              Transposing the $(p+1)$-st row with the $(p+q+1)$-st row, we obtain
              \[
                  \det M = \det \begin{pmatrix}
                      A & B & C \\
                      O & D & O \\
                      O & O & E
                  \end{pmatrix},
              \]
              where
              \begin{align*}
                   & D = \begin{pmatrix}
                      \omega & \omega & \dots  & \omega \\
                      0      & 1      & \dots  & 0      \\
                      \vdots & \vdots & \ddots & \vdots \\
                      0      & 0      & \dots  & 1
                  \end{pmatrix}, \\
                   & E = \begin{pmatrix}
                      \omega^2 & \omega^2 & \dots  & \omega^2 \\
                      0        & 1        & \dots  & 0        \\
                      \vdots   & \vdots   & \ddots & \vdots   \\
                      0        & 0        & \dots  & 1
                  \end{pmatrix}.
              \end{align*}
              Therefore, it follows that $\det M = 1$.
              By the same argument, it holds that if only one of $p,q,r$ is even, then $\det M = 1$.
        \item
              Suppose that $p \equiv 1 \pmod 2$, $q \equiv r \equiv 0 \pmod 2$.
              The value of $\det{M}$ is calculated as follows:
              By the same row operations as in (i), we obtain
              \[
                  \det M = \det \begin{pmatrix}
                      J-A & \omega^2 J-B & \omega J-C \\
                      C   & A            & B          \\
                      B   & C            & A
                  \end{pmatrix}.
              \]
              Add the $(p+1)$-st row of $\det M$ multiplied by $\omega^2$ to the first row.
              Furthermore, add the $(p+q+1)$-st row multiplied by $\omega$ to the first row.
              Consequently, the first row equals to the zero vector.
              Therefore, it follows that $\det M = 0$.
              By the same argument, it holds that if only one of $p,q,r$ is odd, then $\det M = 0$.
    \end{enumerate}
    Therefore, it holds that $\det M=1$ if and only if $\wt(c_i)=p+q+r \equiv 0 \pmod 2$.
    This completes the proof.
\end{proof}

\begin{cor} \label{thm:punctureLCD}
    Let $C$ be a quaternary Hermitian LCD $[n,k]$ code and let $G$ be a generator matrix of $C$ such that $G\overline{G}^T=I_k$.
    Let $l_i$ denote the $i$-st column of $G$.
    Then $C_i$ is a Hermitian LCD code if and only if $\wt(l_i)$ is odd.
\end{cor}
\begin{proof}
    By \cref{thm:P_XOR_S,thm:shortenLCD}, the result follows.
\end{proof}

\section{Optimal quaternary Hermitian LCD codes} \label{sec:optimal}
Let $d_4(n,k)$ denote the largest minimum weight of quaternary Hermitian LCD $[n,k]$ codes.
As stated in \cref{sec:introduction}, we consider $d_4(n,k)$ for $12 \le n \le 30$ and $4 \le k \le n-4$.
In this section, we improve some of the previously known lower bounds on $d_4(n,k)$ for $12 \le n \le 30$ and $4 \le k \le n-4$, using the results of~\cref{sec:construction,sec:puncture}.
The punctured codes and shortened codes in this section were constructed by the {\sc Magma} functions {\tt ShortenCode} and {\tt PunctureCode}, respectively.

In order to improve lower bounds, we use the following method:
Let $d_K(n,k)$ denote the largest minimum weight among currently known $[n,k]$ codes.
By the {\sc Magma} function {\tt BestKnownLinearCode}, one can construct an $[n,k,d_K(n,k)]$ code
for all $12 \le n \le 30$ and $4 \le k \le n-4$.
In addition, considering shortened codes and punctured codes of $[n,k,d_K(n,k)]$ codes, we found quaternary Hermitian LCD $[n,k,d_K(n,k)]$ codes for
\begin{equation}\label{eq:LCD2}
    \begin{split}
        (n,k,d_K(n,k))=&(12,8,4),(13,9,4),(25,20,4),(26,5,16),(26,18,6),\\
        & (26,21,4),(27,5,17),(27,22,4),(28,6,17),(28,23,4),\\
        & (29,24,4),(30,24,4),(30,25,4).
    \end{split}
\end{equation}
In order to reduce the above computation, we used the following method:
Suppose that the punctured (resp. shortened) code of a quaternary Hermitian LCD code $C$ with respect to a coordinate, say $i$, is not a quaternary Hermitian LCD code.
Then, by~\cref{thm:P_XOR_S}, the shortened (resp. punctured) code of $C$ with respect to the coordinate $i$ is a quaternary Hermitian LCD code.
Thus, if an $[n,k,d_K(n,k)]$ code is a quaternary Hermitian LCD code, then we only need to check, for each coordinate, whether the punctured code of $C$ is a quaternary Hermitian LCD code or not.
Consequently, we obtain Proposition~\ref{prp:LCD:trivial}.
\begin{prp} \label{prp:LCD:trivial}
    There exists an optimal quaternary Hermitian LCD $[n,k,d]$ code for $(n,k,d)$ listed in \eqref{eq:LCD2}.
\end{prp}
By a method similar to that given in the above, we found quaternary Hermitian LCD $[n,k,d_K(n,k)-1]$ codes and quaternary Hermitian LCD $[n,k,d_K(n,k)-2]$ codes for
\begin{equation}\label{eq:LCD3}
    \begin{split}
        (n,k,d_K(n,k)-1)=&(14,10,3),(15,11,3),(16,12,3),(17,13,3),(26,4,17),\\
        &(26,17,6),(26,19,5),(26,20,4),(27,18,6),(27,19,5),\\
        &(27,20,5),(27,21,4),(28,5,17),(28,17,7),(28,19,6),\\
        &(28,20,5),(28,21,5),(28,22,4),(29,5,18),(29,6,17),\\
        &(29,18,7),(29,20,6),(29,21,5),(29,22,5),(29,23,4),\\
        &(30,5,19),(30,16,9),(30,19,7),(30,21,6),(30,22,5),\\
        &(30,23,5),
    \end{split}
\end{equation}
\begin{equation}\label{eq:LCD4}
    \begin{split}
        (n,k,d_K(n,k)-2)=&(26,6,14),(26,7,13),(26,8,12),(26,12,9),\\
        &(26,14,7),(26,15,7),(26,16,6),(27,4,17),(27,7,14),\\
        &(27,10,11),(27,13,9),(27,14,8),(27,15,7),(27,16,7),\\
        &(27,17,6),(28,7,14),(28,8,14),(28,10,12),(28,14,9),\\
        &(28,15,8),(28,16,7),(28,18,6),(29,4,18),(29,7,15),\\
        &(29,8,14),(29,9,14),(29,10,13),(29,15,9),(29,16,8),\\
        &(29,17,7),(29,19,6),(30,4,19),(30,6,17),(30,8,15),\\
        &(30,9,14),(30,10,14),(30,17,8),(30,18,7),(30,20,6).
    \end{split}
\end{equation}

Consequently, we obtain Proposition~\ref{prp:LCD:trivial1}.
\begin{prp} \label{prp:LCD:trivial1}
    There exists a quaternary Hermitian LCD $[n,k,d]$ code for $(n,k,d)$ listed in \eqref{eq:LCD3} and~\eqref{eq:LCD4}.
\end{prp}

Furthermore, applying Corollary~\ref{thm:construction_LCD} to quaternary Hermitian LCD $[n,k,d]$ codes for $(n,k,d)$ listed in~\cref{tab:xy}, we found quaternary Hermitian LCD $[n,k,d]$ codes for
\begin{equation}\label{eq:LCD1}
    \begin{split}
        (n,k,d)=&(27,6,15),(27,8,13),(27,9,12),(27,11,10),(28,11,11),\\
        &(29,11,11),(30,7,16),(30,12,11).
    \end{split}
\end{equation}
Therefore, we obtain Proposition~\ref{prp:LCD:PI}.
\begin{prp} \label{prp:LCD:PI}
    There exists a quaternary  LCD $[n,k,d]$ code for $(n,k,d)$ listed in \eqref{eq:LCD1}.
\end{prp}

The vectors $x,y$ in Corollary~\ref{thm:construction_LCD} are listed in~\cref{tab:xy}.
For each of the parameters given in~\eqref{eq:LCD2} through~\eqref{eq:LCD1}, a quaternary Hermitian LCD code can be obtained electronically from \url{https://www.math.is.tohoku.ac.jp/~mharada/Ishizuka/F4_HLCD_genmat.txt}.
Note that, the generator matrix used to construct a quaternary Hermitian LCD code in~\cref{prp:LCD:PI} can be obtained from $x,y$ in~\cref{tab:xy} and the generator matrix given in~\url{https://www.math.is.tohoku.ac.jp/~mharada/Ishizuka/F4_HLCD_genmat.txt}.
\begin{landscape}
    \begin{table}[thbp]\caption{Vectors $x,y$ applied to quaternary Hermitian LCD $[n,k,d]$ codes}\label{tab:xy}
        \begin{center}
            \begin{tabular}{c|ccccccccccccccccccccccccccccccccccc}
                $(n,k,d)$    & $x$                          & $y$                          \\
                \hline
                $(27,6,14)$  & $\begin{smpmat}\omega&1&\omega^2&1&1&\omega&\omega&1&\omega&0&0&0&\omega^2&0&\omega&0&\omega&\omega&\omega&\omega^2&\omega\\ \end{smpmat}$ & $\begin{smpmat}\omega&\omega&\omega&\omega&\omega&\omega^2&0&\omega^2&1&\omega^2&0&0&0&0&\omega&1&\omega^2&\omega&\omega^2&\omega&\omega\\ \end{smpmat}$ \\
                $(27,8,12)$  & $\begin{smpmat}\omega^2&0&\omega&1&1&\omega^2&0&\omega^2&\omega^2&0&\omega&1&0&\omega^2&\omega^2&\omega&\omega^2&1&0\\ \end{smpmat}$ & $\begin{smpmat}\omega^2&\omega^2&1&\omega^2&\omega^2&\omega&\omega^2&1&1&\omega&1&0&0&1&\omega&0&0&\omega^2&0\\ \end{smpmat}$ \\
                $(27,9,11)$  & $\begin{smpmat}1&1&1&\omega&0&1&0&\omega&0&\omega&0&\omega^2&\omega&1&\omega&1&\omega^2&\omega\\ \end{smpmat}$ & $\begin{smpmat}0&0&\omega^2&1&\omega&\omega&\omega&0&1&0&\omega&1&1&\omega^2&1&\omega^2&0&0\\ \end{smpmat}$ \\
                $(27,11,9)$  & $\begin{smpmat}\omega^2&1&1&\omega^2&\omega&0&\omega^2&\omega&\omega^2&0&\omega&\omega^2&1&1&\omega&\omega^2\\ \end{smpmat}$ & $\begin{smpmat}0&0&0&1&1&\omega^2&1&\omega&1&0&1&\omega^2&1&0&1&0\\ \end{smpmat}$ \\
                $(28,11,10)$ & $\begin{smpmat}\omega&\omega^2&0&\omega&\omega&0&\omega&1&\omega^2&0&0&0&1&1&\omega^2&\omega&1\\ \end{smpmat}$ & $\begin{smpmat}0&\omega^2&\omega^2&1&0&0&1&\omega&\omega^2&1&\omega&1&\omega&0&0&\omega^2&\omega\\ \end{smpmat}$ \\
                $(29,11,10)$ & $\begin{smpmat}0&\omega&0&1&0&\omega&\omega&\omega^2&1&\omega&\omega^2&1&0&\omega^2&0&\omega^2&\omega^2&0\\ \end{smpmat}$ & $\begin{smpmat}\omega^2&\omega^2&\omega&\omega&\omega^2&\omega&0&1&\omega^2&0&\omega^2&0&0&1&1&\omega&0&0\\ \end{smpmat}$ \\
                $(30,7,15)$  & $\begin{smpmat}\omega^2&\omega&1&1&0&\omega^2&\omega^2&\omega&1&1&\omega^2&0&0&\omega^2&1&\omega^2&\omega^2&\omega^2&0&\omega&\omega&0&\omega^2\\ \end{smpmat}$ & $\begin{smpmat}1&\omega&1&0&0&1&\omega&1&1&\omega&1&0&1&0&\omega&0&\omega^2&\omega^2&1&\omega^2&\omega&0&0\\ \end{smpmat}$ \\
                $(30,12,10)$ & $\begin{smpmat}1&0&0&1&\omega&0&\omega&0&\omega^2&\omega^2&\omega^2&\omega^2&1&\omega&0&0&\omega&1\\ \end{smpmat}$ & $\begin{smpmat}1&\omega^2&0&0&\omega^2&\omega&0&\omega^2&0&\omega^2&\omega^2&0&\omega&0&1&\omega^2&\omega^2&\omega^2\\ \end{smpmat}$ \\
            \end{tabular}
        \end{center}
    \end{table}
\end{landscape}

As stated earlier, for $n \le 20$ and $4 \le k \le 20$, $d_4(n,k)$ is given with an upper bound in~\cite{harada_optimal_2019} and~\cite{lu_optimal_2020}.
For a quaternary Hermitian LCD $[n,k]$ code with $21 \le n \le 30$ and $4 \le k \le n-4$, we used the following to obtain upper bounds:
\begin{equation*}
    d_4(n,k) \le d(n,k),
\end{equation*}
where $d(n,k)$ denotes the largest minimum weight among all $[n,k]$ codes.
The values of $d(n,k)$ are given in~\cite{grassl_bounds_nodate}.

In~\cref{tab:F4-0,tab:F4-1}, we give $d_4(n,k)$ for $12 \le n \le 30$ and $4 \le k \le n-4$.
For the parameters listed in Proposition~\ref{prp:LCD:PI}, we mark $d_4(n,k)$ by $*$ in~\cref{tab:F4-0,tab:F4-1}.
Furthermore, for the parameters listed in Propositions~\ref{prp:LCD:trivial} through~\ref{prp:LCD:PI}, we give $d_4(n,k)$ in boldface.

\section*{Acknowledgments}
The author would like to thank supervisor Professor Masaaki Harada for his helpful advice and encouragement.
The author also would like to thank Ken Saito for his kind permission to publish his result in~\cref{thm:orthnormal_basis}.

\begin{table}[thbp]\caption{$d_{4}(n,k)$ for $4 \le k \le 12$}\label{tab:F4-0}\begin{center}{\small\begin{tabular}{c|cccccccccccccccccccccccccccccccccccccccccccccccc}\noalign{\hrule height 0.8pt}
                $n\backslash k$ & 4               & 5               & 6                & 7                & 8                & 9                & 10              & 11               & 12               \\
                12              & 7               & 6               & 5--6             & 5                & \textbf{4}       &                  &                 &                  &                  \\
                13              & 8               & 7               & 6                & 5                & 4                & \textbf{4}       &                 &                  &                  \\
                14              & 8               & 7--8            & 7                & 6                & 5                & 4                & \textbf{3}--4   &                  &                  \\
                15              & 9               & 8               & 7                & 7                & 6                & 5                & 4               & \textbf{3}--4    &                  \\
                16              & 10              & 9               & 8                & 7--8             & 6--7             & 6                & 5               & 5                & \textbf{3}--4    \\
                17              & 11              & 9--10           & 9                & 8                & 7--8             & 6--7             & 6               & 5--6             & 4                \\
                18              & 11--12          & 10--11          & 9--11            & 9                & 8--9             & 7--8             & 6--7            & 5--6             & 5                \\
                19              & 12--13          & 11              & 10--11           & 9                & 8--9             & 8                & 7               & 6--7             & 5--6             \\
                20              & 13              & 12              & 11--12           & 10               & 8--9             & 8--9             & 7--8            & 6--7             & 6--7             \\
                21              & 14              & 12              & 12               & 10--11           & 9--10            & 8--9             & 8--9            & 7--8             & 7--8             \\
                22              & 14              & 13              & 12--13           & 11--12           & 10               & 8--10            & 8--9            & 7--9             & 7--9             \\
                23              & 15              & 14              & 13               & 12--13           & 10--12           & 9--11            & 8--10           & 8--9             & 8--9             \\
                24              & 16              & 15              & 14               & 12--13           & 11--13           & 10--12           & 9--11           & 8--10            & 8--10            \\
                25              & 17              & 15              & 14--15           & 13--14           & 12--13           & 11--13           & 11--12          & 9--11            & 9--11            \\
                26              & \textbf{17}--18 & \textbf{16}     & \textbf{14}--16  & \textbf{13}--15  & \textbf{12}--14  & \textbf{10}--13  & \textbf{10}--13 & \textbf{9}--12   & \textbf{9}--11   \\
                27              & \textbf{17}--19 & \textbf{17}     & \textbf{15*}--16 & \textbf{14}--16  & \textbf{13*}--15 & \textbf{12*}--14 & \textbf{11}--13 & \textbf{10*}--13 & \textbf{9}--12   \\
                28              & \textbf{17}--20 & \textbf{17}--18 & \textbf{17}      & \textbf{14}--16  & \textbf{14}--16  & \textbf{12}--15  & \textbf{12}--14 & \textbf{11*}--13 & \textbf{9}--12   \\
                29              & \textbf{18}--20 & \textbf{18}--19 & \textbf{17}--18  & \textbf{15}--17  & \textbf{14}--16  & \textbf{14}--16  & \textbf{13}--15 & \textbf{11*}--14 & \textbf{10}--13  \\
                30              & \textbf{19}--21 & \textbf{19}--20 & \textbf{17}--19  & \textbf{16*}--18 & \textbf{15}--17  & \textbf{14}--16  & \textbf{14}--16 & \textbf{12}--15  & \textbf{11*}--14 \\
                \noalign{\hrule height 0.8pt}\end{tabular}}\end{center}\end{table}

\begin{table}[thbp]\caption{$d_{4}(n,k)$ for $13 \le k \le 26$}\label{tab:F4-1}\begin{center}{\small\begin{tabular}{c|cccccccccccccccccccccccccccccccccccccccccccccccc}\noalign{\hrule height 0.8pt}
                $n\backslash k$ & 13              & 14             & 15             & 16             & 17             & 18            & 19            & 20            & 21            & 22            & 23            & 24         & 25         & 26 \\
                17              & 3--4            &                &                &                &                &               &               &               &                                                                              \\
                18              & 4               & 3              &                &                &                &               &               &               &                                                                              \\
                19              & 5               & 4              & 3              &                &                &               &               &               &                                                                              \\
                20              & 5--6            & 5              & 4              & 3              &                &               &               &               &                                                                              \\
                21              & 6               & 5--6           & 5              & 4              & 3              &               &               &               &                                                                              \\
                22              & 6--7            & 6              & 5--6           & 4--5           & 4              & 3             &               &               &                                                                              \\
                23              & 6--8            & 6--7           & 6              & 5--6           & 4--5           & 4             & 3             &               &                                                                              \\
                24              & 7--9            & 6--8           & 6--7           & 6              & 5--6           & 4--5          & \textbf{4}--3 & 3             &                                                                              \\
                25              & 7--9            & 7--9           & 6--8           & 6--7           & 6              & 5--6          & 4             & \textbf{4}    & 3                                                                            \\
                26              & 8--10           & \textbf{7}--9  & \textbf{7}--9  & \textbf{6}--8  & \textbf{6}--7  & \textbf{6}    & \textbf{5}--6 & \textbf{4}--5 & \textbf{4}    & 3                                                            \\
                27              & \textbf{9}--11  & \textbf{8}--10 & \textbf{7}--9  & \textbf{7}--9  & \textbf{6}--8  & \textbf{6}--7 & \textbf{5}--6 & \textbf{5}--6 & \textbf{4}--5 & \textbf{4}    & 3             &            &            &    \\
                28              & \textbf{9}--12  & \textbf{9}--11 & \textbf{8}--10 & \textbf{7}--9  & \textbf{7}--8  & \textbf{6}--8 & \textbf{6}--7 & \textbf{5}--6 & \textbf{5}--6 & \textbf{4}--5 & \textbf{4}    & 3          &            &    \\
                29              & \textbf{9}--12  & \textbf{9}--12 & \textbf{9}--11 & \textbf{8}--10 & \textbf{7}--9  & \textbf{7}--8 & \textbf{6}--8 & \textbf{6}--7 & \textbf{5}--6 & \textbf{5}--6 & \textbf{4}--5 & \textbf{4} & 3          &    \\
                30              & \textbf{10}--13 & \textbf{9}--12 & \textbf{9}--12 & \textbf{9}--10 & \textbf{8}--10 & \textbf{7}--9 & \textbf{7}--8 & \textbf{6}--8 & \textbf{6}--7 & \textbf{5}--6 & \textbf{5}--6 & \textbf{4} & \textbf{4} & 3  \\
                \noalign{\hrule height 0.8pt}\end{tabular}}\end{center}\end{table}

\bibliography{main}

\end{document}